\documentclass[letterpaper, 10 pt, conference]{ieeeconf}

\IEEEoverridecommandlockouts

\usepackage{amssymb,latexsym,amsfonts,amsmath}
\usepackage{theorem}%
\usepackage{graphicx,color}%
\usepackage{tikz}
\usepackage{hyperref}  
\usepackage{enumerate}
\usepackage{cite}
\allowdisplaybreaks
\newtheorem{definition}{Definition}[section]%
\newtheorem{theorem}[definition]{Theorem}%
\newtheorem{proposition}[definition]{Proposition}%
\newtheorem{assumption}[definition]{Assumption}%
\newtheorem{remark}[definition]{Remark}

\newcommand{\tm}{\times}%
\newcommand{\trn}{^{\scriptscriptstyle \top}}%

\def\field#1{\mathbb #1}%
\def\R{\field{R}}%

\newcommand{\N}{\mathbb{N}}%
\newcommand{\UC}{\mathcal{U}}%
\newcommand{\diag}{\mathrm{diag}}%
\DeclareMathOperator{\id}{id}

\def\K{\mathcal{K}}%
\def\Kinf{\mathcal{K}_\infty}%
\let\ol=\overline%
\let\ul=\underline%

\title{\huge Compositional Construction of Abstractions for Infinite Networks of Switched Systems}%

\author{Maryam Sharifi\thanks{M.~Sharifi is with the School of Electrical and Computer Engineering, University of Tehran, Iran; e-mail: \texttt{sharifi.m@ut.ac.ir}.},
	Abdalla~Swikir\thanks{A.~Swikir is with Department of Electrical and Computer Engineering, Technical University of Munich, Germany; e-mail: \texttt{abdalla.swikir@tum.de}.},
	Navid~Noroozi\thanks{N.~Noroozi is with the Institute of Informatics, LMU Munich, Germany; e-mail: \texttt{navid.noroozi@lmu.de}. His work is supported by the DFG through the grant WI 1458/16-1.}, and
	Majid Zamani\thanks{M.~Zamani is with the Computer Science Department, University of Colorado Boulder, CO 80309, USA. 
		M.~Zamani is also with the Institute of Informatics, LMU Munich, Germany;  email: {\tt\small majid.zamani@colorado.edu}. His work is supported in part by the DFG through the grant ZA 873/4-1 and the H2020 ERC Starting Grant AutoCPS (grant agreement No.~804639).}
}

\begin{document}

	
	\maketitle%
	
	\begin{abstract}
		We construct compositional continuous approximations for an interconnection of infinitely many discrete-time switched systems.
		An approximation (known as abstraction) is itself a continuous-space system, which can be used as a replacement of the original (known as concrete) system in a controller design process.
		Having synthesized a controller for the abstract system, the controller is refined to a more detailed controller for the concrete system.
		To quantify the mismatch between the output trajectory of the approximation and of that the original  system, we use the notion of so-called simulation functions.
		In particular, each subsystem in the concrete network and its corresponding one in the abstract network is related through a local simulation function.
		We show that if the local simulation functions satisfy a certain small-gain type condition developed for a network of infinitely many subsystems, then the aggregation of the individual simulation functions provides an overall simulation function between the overall abstraction and the concrete network.
		For a network of linear switched systems, we systematically construct local abstractions and local simulation functions, where the required conditions are expressed in terms of linear matrix inequalities and can be efficiently computed.
		We illustrate the effectiveness of our approach through an application to frequency control in a power gird with a switched (i.e. time-varying) topology.
	\end{abstract}

	\section{introduction}
%
The high cost of incorrect configuration of a control system, on one hand, and safety concerns, on the other hand, call for automated and provably correct techniques for the verification and synthesis of modern control systems.
In addition, emergent applications which consist of large-scale networked systems such as smart grids, connected automated vehicles, swarm robotics, etc. necessitate advanced control objectives going well beyond classic control problems such as regulation and tracking.

The complexity of control objectives, large and time-varying number of participating agents, and the complexity of the problem require methods on automated synthesis of provably correct controllers by joining forces from control theory and computer science.
Particularly, a discrete abstraction (refereed to as symbolic model) provides automated synthesis of a correct-by-design controller for the original (referred to as concrete) system. 
In this approach, the controller synthesis problem can be algorithmically solved over a finite abstraction of the concrete system.
Then, the constructed controller is refined back to the original system based on some behavioral relation between the original system and its finite abstraction such as approximate alternating simulation relations~\cite{pt09}.

	
	The applicability of finite abstractions is considerably limited due to the computational complexity of constructing discrete approximations of the concrete system.
	Therefore, a brute force approach to large-scale systems is not feasible.
	A  way to reduce this
	computational complexity is to introduce a pre-processing step by constructing so-called \emph{continuous} abstractions.
	In that way, a continuous-space system, but possibly with a \emph{lower} dimension, is obtained for the concrete system~\cite{girard2009hierarchical,Smith.2018,Smith.2019}.
	To further manage the computational complexity, one may divide a possibly large-scale network into several smaller subsystems and then construct an abstraction for each subsystem individually.
	The methodology to achieve an abstraction for the overall network via the interconnection of the individual abstractions is called a compositional approach~\cite{rungger2016compositional,ZamaniArcak2017,Noroozi.2018b}.
	However, an efficient approach which is \emph{independent} of the size of the network and potentially applicable to infinite-dimensional cases is still missing.
	
	Motivated by the above discussion, this paper aims at providing a \emph{scale-free} compositional approach for the construction of continuous abstractions for \emph{arbitrarily}	 large-scale networks of discrete-time switched systems.
	Inspired by the works in the literature regarding stability analysis of large-scale systems e.g.~\cite{NMK20b,DaP20,Bamieh.2012,Barooah.2009,BPD02}, to address the scalability issue, we over-approximate a finite-but-large network with a network composed of \emph{infinitely} many subsystems, which we call it an infinite network.
	It is widely accepted that an infinite network {captures} the essence of its corresponding finite network; see, e.g., a vehicle platooning application in~\cite{Jovanovic.2005b}.
	This treatment leads to an infinite-dimensional system and calls for a more rigorous and detailed setting.
	In particular, we adapt the notion of simulation functions~\cite{girard2009hierarchical} to the case of \emph{infinite}-dimensional systems.
	The existence of a simulation function ensures that the error between the output trajectories of the abstract system and that of the concrete system is quantitatively bounded in a certain sense (cf. Definition~\ref{def:sim-Function}).

	Following the compositionality approach, we assign to each subsystem an individual simulation function and construct each local abstraction accordingly.
	Then we aggregate them to compose an abstraction for the overall network.
	We show that the aggregation yields a continuous abstraction for the overall concrete network if a certain small-gain condition, which has been recently developed in~\cite{kawan2019lyapunov}, is satisfied.
	The effectiveness of our approach is verified by an application to frequency control in a power grid  with a time-varying topology.
	
	{\it Notation:} We write $\N_0 (\N)$ for the set of nonnegative (positive) integers.
	For vector norms on finite- and infinite-dimensional vector spaces, we write $|\cdot|$.
	By $\ell^p$, $p\in[1,\infty)$, we denote the Banach space of all real sequences $x = (x_i)_{i\in\N}$ with finite $\ell^p$-norm $|x|_p<\infty$, where $|x|_p = (\sum_{i=1}^{\infty}|x_i|^p)^{1/p}$ for $p < \infty$.
	If $X$ is a Banach space, we write $r(T)$ for the spectral radius of a bounded linear operator $T:X\to X$. 
	The identity function is denoted by $\id$.
	We will consider $\K$ and $\Kinf$ comparison functions, see~\cite[Chapter 4.4]{Khalil.2002} for definitions.
	
	\section{System Description}
	
		We study the interconnection of countably many switched systems, each given by a finite-dimensional difference equation. We define the switching signal functions ${\sigma_i}:\N_0 \to S_i$, ${i\in\N}$ for $S_i\in\{1,2,\dots,r_i\}$ which is a finite index set with $r_i\in\N$. We denote the set of such switching signals by ${\cal{S}}_i$. The $i$-th subsystem ($i\in\N$) is written as%
	\begin{equation}\label{eq_ith_subsystem}
	{\Sigma _i}:\quad \left\{ \begin{array}{l}
	\mathbf x_i(k+1)  = {f_{i,{\sigma_i(k)}}}({\mathbf x_i(k)},{\mathbf w_i(k)},{\mathbf u_i(k)}),\\
	{\mathbf y_i(k)} = h_{i, \sigma_i(k)}(\mathbf x_i(k)),
	\end{array} \right.%
	\end{equation}
	where $\mathbf x_i:\N_0 \rightarrow  \R^{n_i}$, $\mathbf w_i:\N_0 \rightarrow \R^{N_i}$, $\mathbf u_i:\N_0 \rightarrow \R^{m_i}$, and $\mathbf{y}_i:\N_0\rightarrow  \R^{q_i}$ are state signal, internal input signal, external input signal, and output signal, respectively.
	
	The family $(\Sigma_i)_{i\in\N}$ comes together with sequences $(n_i)_{i\in\N}$, $(m_i)_{i\in\N}$ of positive integers and finite sets $I_i \subset \N \backslash \{i\}$, and $\overline{I}_i \subset \N$ enumerate the neighbors of $\Sigma_i$, i.e., those systems $\Sigma_j,j \in I_i$, $\Sigma_{j'},j' \in \overline{I}_i$ that affect or are affected by $\Sigma_i$, respectively. By definition we require that $i\notin I_i\cup\overline{I}_i$, $\forall i\in \N$.
		We denote $\mathbf w_i(k)= \left( {{\mathbf w}_{ij}(k)}\right)_{j\in I_i}\in \R^{N_i}$ for $N_i := \sum_{j\in I_i}n_j$ as the internal inputs to show the interconnections. 
	The output functions ${h_{i, \sigma_i(k)}}({\mathbf x_i(k)})=\left( {h_{ij,\sigma_i(k)}(\mathbf x_i(k))}\right)_{j\in (i\cup\overline{I}_i)}$, $\mathbf y_i(k)= \left({{\mathbf y}_{ij}(k)}\right)_{j\in (i\cup\overline{I}_i)}$ are elements of $\R^{q_i}$. 
	Note that $\mathbf w_i(k)$ and $\mathbf y_i(k)$ are partitioned into sub-vectors and we aggregate all the subsystems $\Sigma_i$ through the interconnection constraints given by $ \mathbf w_{ij}(k)=\mathbf y_{ji}(k)$ for all $i\in\N$ and for all $j\in I_i$.
	In that way, the interconnection of $\Sigma _i$, $i\in \N$, is described by 
		\begin{equation}\label{eq_interconnection}
		{\Sigma}:\quad \left\{ \begin{array}{l}
		\mathbf x(k+1)  = {f_{\sigma(k)}}({\mathbf x(k)},{\mathbf u(k)}),\\
		{\mathbf y(k)} = h_{\sigma(k)}(\mathbf x(k)),
		\end{array} \right.%
		\end{equation}
		where $\mathbf x(k) \!\!=\!\! (\mathbf x_i(k))_{i\in\N}$,  $\mathbf u(k) \!\!=\!\! (\mathbf u_i(k))_{i\in\N}$, $\mathbf y(k) = (\mathbf y_{ii}(k))_{i\in\N}$, $\sigma(k) \!\!=\!\! (\sigma_i(k))_{i\in\N}$, $f_{\sigma(k)}(\mathbf x(k),\mathbf u(k)) = \left(f_{i,\sigma_i(k)}(\mathbf x_i(k),\mathbf w_i(k),\mathbf u_i(k))\right)_{i\in\N}$, and $h_{\sigma(k)}(\mathbf x(k)):=\left(h_{ii,\sigma_i(k)}(\mathbf x_i(k))\right)_{i\in\N}$.
	
	Clearly, system~\eqref{eq_interconnection} is an \emph{infinite}-dimensional system, which asks for careful choice of the state and input spaces. We choose appropriate Banach spaces $X \subset \prod_{i\in\N}\R^{n_i}$  and $U \subset \prod_{i\in\N}\R^{m_i}$, and restrict $f_{\sigma(k)}$ to $X \tm U$, $\sigma:\N \to S$, for all $k\in\N_0$, where $S=\prod_{i\in\N}S_i$.
	
	We  model the state space $X$ of $\Sigma$ as a Banach space of sequences $x = (x_i)_{i\in\N}$ with $x_i \in \R^{n_i}$.
	The most natural choice is an $\ell^p$-space. To define such a space, we first fix a norm on each $\R^{n_i}$.
	Then, for every $p \in [1,\infty)$, we put%
	\begin{equation*}
	\ell^p(\N,(n_i)) := \Bigl\{x = (x_i)_{i\in\N} : x_i \in \R^{n_i},\ \sum_{i\in\N}|x_i|^p < \infty \Bigr\}%
	\end{equation*}
	and equip this space with the norm $|x|_p := (\sum_{i\in\N}| x_i|^p)^{1/p}$.
	
	As the state space of the system $\Sigma$, we consider $X := \ell^p(\N,(n_i))$ for a fixed $p \in [1,\infty)$.
	Similarly, for a fixed $q \in [1,\infty)$, we consider the \emph{external input space} $U := \ell^q(\N,(m_i))$, 
	where we fix norms on $\R^{m_i}$ that we simply denote by $|\cdot|$ again.
	The space of admissible \emph{external input functions} $\mathbf u$ is defined by $\UC := \bigl\{\mathbf u:\N_0 \rightarrow U  \bigr\}$.
	We denote the corresponding solutions by $\mathbf x(k,x,\sigma,\mathbf u)$ for any $k\in\N_0$, any initial value $x\in X$, any switching signal $\sigma\in\cal{S}$, and any control input $\mathbf u\in\UC$. 

	We refer to system \eqref{eq_interconnection} as the \emph{concrete} system, which is often hard to control. In order to synthesize these systems, using a simpler, though less precise system called an \emph{abstract} system is beneficial.
	Adopting the same \emph{notational convention} as those for $\Sigma_i$ and $\Sigma$, but  with the $\hat{\cdot}$ sign on the top of the respective ones, we introduce the symbols for the abstract subsystems  $\hat \Sigma_i$ and the corresponding overall system and $\hat \Sigma$, respectively.
	\section{Abstractions for switched discrete-time systems}\label{sec:SIM}
	In this section, we introduce the notion of so-called simulation functions for the discrete-time switched systems with only external inputs.
	A simulation function of $\hat\Sigma$ by $\Sigma$ is a function over their state spaces which explain how a state trajectory of $\hat\Sigma$ can be transformed into a state trajectory of $\Sigma$ to make the distance between the associated output trajectories bounded. 
	Formally, a simulation function is defined as follows:
	\begin{definition}\label{def:sim-Function}
		Consider the systems $\Sigma$ and $\hat \Sigma$ with the same output spaces and fixed $p,q \in [1,\infty)$. Let $V_s:X\times \hat X \rightarrow \R_+, s\in S$, be a family of functions. Let there exist positive constants $\alpha, b$, such that for all $s \in S$, $x \in X$, $\hat x \in \hat X$, 
		\begin{align}\label{sim2}
		&\alpha \left| {h_s(x) - \hat h_s(\hat x)} \right|_p^b \le V_s(x,\hat x),
		\end{align}
		and there exist a function ${\rho _{\rm ext}}\in\K$ and a positive constant $\lambda<1$ , such that for all consecutive $s',s\in S$ (i.e., $s'=\sigma(k+1), s=\sigma(k)$ for $k\in \N_0$), and all $x \in X$, $\hat x \in \hat X$ and $\hat u\in \hat U$ there exist  $u \in U$ so that we have
		\begin{align}\label{sim3}
		&\begin{array}{l}
		V_{s'}(f_{s}(x,u),\hat f_{s}(\hat x,\hat u)) - V_{s}(x,\hat x)\\\leq
		- \lambda V_{s}(x,\hat x) + {\rho _{\rm ext}}(|\hat u|_{q}).
		\end{array} %
		\end{align}
		Then, the functions $V_s$ are called the simulation functions from $\hat \Sigma$ to $\Sigma$.%
	\end{definition}
	The following proposition shows the importance of the existence of a simulation function.
	\begin{proposition}
		Consider systems $\Sigma$ and $\hat\Sigma$,  the same output space, and fixed $p,q \in [1,\infty)$. Let a set of simulation functions $V_s$, $s\in S$, from $\hat\Sigma$ to $\Sigma$ be given. Then there exist a function $\gamma_{\rm{ext}}\in\K$ and positive constants $\vartheta$ and $\beta<1$, such that for any $\sigma\in\cal{S}$, $x\in X$, $\hat x \in \hat X$, $\mathbf{\hat u}\in \hat \UC$, $k\in\N_0$, there exists $\mathbf u\in \UC$ so that we have
		\begin{align}\label{ms}
		&\left| {\mathbf y(k,x,\sigma,\mathbf u) -  \mathbf {\hat y}(k,\hat x,\sigma,\mathbf{\hat u})} \right|_p \nonumber\\&\le \vartheta\beta^k (V_{\sigma(0)}(\xi,\hat \xi))^{\frac{1}{b}}+\gamma_{\rm ext}(|\mathbf{\hat u}|_{q,
			\infty}),
		\end{align}	
		where $|\mathbf{\hat u}|_{q,\infty} := \sup_{k\in \N_0}|\mathbf{\hat u}(k)|_q$ and $b$ as in~\eqref{sim2}.

	\end{proposition}
		The proof is not presented due to space limitations. Basically it follows similar arguments as those in the proof of~\cite[Lemma 3.5]{Jiang.2001}.
	\begin{remark}\label{input}
		If we are given an interface function $\nu$ that maps
		every $x, \hat{x}$, $\hat{u}$, and $s$ to an input $u = \nu(x, \hat{x},\hat{ u},s)$ so that \eqref{sim3} is
		satisfied, then, the input $\mathbf{u}$ that realizes \eqref{ms} is readily given
		by $\mathbf{u}(k)=\nu(\mathbf x(k), \hat{\mathbf x}(k),\hat{ \mathbf u}(k),\sigma(k))$, see \cite[Theorem 1]{Runggerhscc}.
	\end{remark}
	\section{Compositional Construction of Abstractions and Simulation Functions}\label{sec:SIM-cons}
	In this section, we construct continuous compositional abstraction for an interconnection of countably many discrete-time switched system and the corresponding simulation function from the abstractions of the subsystems and their corresponding simulation functions, respectively. Then, we focus on linear  subsystems and provide conditions under which local quadratic simulation functions with their associated interface functions construct the abstractions. 
	
	We assume that subsystems $\Sigma_i$ for $i \in \N$, given by~\eqref{eq_ith_subsystem}, together with their abstractions $\hat \Sigma_i$ and the there exist simulation functions $V_{i,s_i}, s_i\in S_i$, from $\hat \Sigma_i$ to $\Sigma_i$ satisfying the following assumption
	\begin{assumption}\label{SIM_vi_existence}
		Consider the subsystems $\Sigma_i$ for $i \in \N$, together with their abstractions $\hat \Sigma_i$. For fixed $p,q \in [1,\infty)$, there exist functions $V_{i,s_i}:\R^{n_i}\times \R^{\hat n_i} \rightarrow \R_+, s_i\in S_i$, with the following properties.%
		\begin{itemize}
			\item There are positive constants $\alpha_i$ so that for all $x_i \in \R^{n_i}$, all $\hat x_i \in \R^{\hat n_i}$%
			\begin{equation}\label{sim_eq_viest}
			\alpha_i \left| {h_{i, s_i}(x_i) - \hat h_{i, s_i}(\hat x_i)} \right|^p \le V_{i,s_i}(x_i,\hat x_i).%
			\end{equation}
			\item  There are positive constants $\lambda_i<1,\rho_{i,\rm int}, \rho_{\rm{i,ext}}$, such that for all consecutive  $s'_i, s_i\in S_i$, $x_i \in \R^{n_i}$, $\hat x_i \in \R^{\hat n_i}$, $\hat u_i\in\R^{\hat m_i}$, there exist $u_i\in\R^{m_i}$, so that the following holds for all $w_i \in \R^{N_i}$, $\hat w_i \in \R^{\hat N_i}$: 
			\begin{align}\label{sim_eq_nablaviest}	
			\begin{array}{l}
			V_{i, s'_i}\left(f_{i,s_i}(x_i, w_i, u_i),\hat f_{i,s_i}(\hat x_i, \hat w_i, \hat u_i)\right)- V_{i, s_i}\left(x_i,\hat x_i\right)\\\le
			- \lambda_i V_{i, s_i}(x_i,\hat x_i) + {\rho _{\rm{i,ext}}}|\hat u_i|^{q} + {\rho _{{\mathop{i,\rm int}} }}\left| { w_i - {\hat w_i}} \right|^p.
			\end{array}%
			\end{align}	
		\end{itemize}
	\end{assumption}
	We assume that the following uniformity conditions hold for the constants introduced above.%
	
	\begin{assumption}\label{ass_external_gains}
		There are constants $\underline{\alpha}, \underline{\lambda}, {{\overline\rho _{\rm{ext}}}} > 0$ so that for all $i\in \N$, we have $\ul\alpha \leq \alpha_i , \ul\lambda \leq \lambda_i, {\rho _{\rm{i,ext}}} \leq {{\overline\rho _{\rm{ext}}}}$.
	\end{assumption}
	In order to formulate a small-gain condition, we further introduce the following matrices by utilizing the coefficients from \eqref{sim_eq_nablaviest}:
	\begin{align}\label{condd1}
	&\Lambda := \diag(\lambda_1,\lambda_2,\lambda_3,\ldots), \;\;\Gamma := (\gamma_{ij})_{i,j\in\N},
	\end{align}
	where
	\begin{align}\label{gamma}
	&{\gamma _{ij}}: = \left\{ {\begin{array}{*{20}{c}}
		{{\rho _{i,{\rm{int}}}}{{\bar N}_i}\frac{1}{{{\alpha _j}}},}& j \in I_i ,\\
		{0,}& j \notin I_i ,
		\end{array}} \right.
	\end{align}
	for ${\bar N}_i$ as the number of neighbors of subsystem $i$.
	Now, we define the following matrix by which we express our \emph{small-gain} condition:
	\begin{align}
	\Psi := \Lambda^{-1} \Gamma := (\psi_{ij})_{i,j\in\N} , \,\,\, \psi_{ij} = \gamma_{ij}/\lambda_i .
	\end{align}
	We make the following spectral radius condition which provides a \emph{quantitative} index on the strength of coupling between the subsystems.
	\begin{assumption}\label{ass:spectral-radius}
		The spectral radius $r(\Psi) < 1$.
	\end{assumption}
	We make an assumption on the boundedness of the operator $\Gamma$.
	\begin{assumption}\label{SIM_def}
		The operator $\Gamma = (\gamma_{ij})_{i,j\in\N}$ satisfies
		$\sup_{j \in \N} \sum_{i=1}^{\infty} \gamma_{ij} < \infty$.
	\end{assumption}	
	Note that the assumption above holds if each subsystem is interconnected to a finitely many subsystems.
	
	The following theorem gives the \emph{main result} of the paper, which is a compositional approach for construction of abstractions of infinite interconnected control systems and their corresponding simulation functions.
	\begin{theorem}\label{MTC}
		Consider the infinite networks $\Sigma$ and $\hat\Sigma$ with fixed $p,q \in [1,\infty)$. Suppose that Assumptions~\ref{SIM_vi_existence},~\ref{ass_external_gains},~\ref{ass:spectral-radius} and \ref{SIM_def} hold.
		Then there exists a vector $\mu = (\mu_{i})_{i\in\N}\in \ell^{\infty}$ satisfying $\underline{\mu} \leq \mu_i \leq \overline{\mu}$ with some constants $\underline{\mu},\overline{\mu}>0$, such that the following is satisfied for a constant $0<\lambda_{\infty}<1$.
		\begin{equation}\label{eq_muest2}
		\frac{[\mu\trn(-\Lambda + \Gamma)]_i}{\mu_{i}} \leq -\lambda_{\infty} \quad \forall i \in \N, s\in S.%
		\end{equation}
		Moreover, for all $s_i\in S_i$, $s\in S$,
		\begin{equation*}\label{eq:sim-function-construction}
		V_s(x, \hat x) = \sum_{i=1}^{\infty} \mu_{i} V_{i,s_i}(x_i, \hat x_i),\quad V_s:X\times \hat X \rightarrow \R_+  %
		\end{equation*}
		are simulation functions of $\hat \Sigma$ by $\Sigma$ with $b = p , \alpha = \ul\mu	\ul\alpha$ as in~\eqref{sim2} and $\lambda = \lambda_\infty$ and $\rho_{\mathrm{ext}} : t \mapsto \overline{\mu}\; \ol\rho _{\rm{ext}} t^q$ as in~\eqref{sim3}.
	\end{theorem}
	\begin{proof}
		From~\cite[Lemma V.10]{kawan2019lyapunov}, Assumption~\ref{ass:spectral-radius} (i.e. $r(\Psi) < 1$) implies that there exists a vector $\mu = (\mu_{i})_{i\in\N}\in \ell^{\infty}$ satisfying $\underline{\mu} \leq \mu_i \leq \overline{\mu}$ such that~\eqref{eq_muest2} holds.
		
		Now we show that $V$ in~\eqref{eq:sim-function-construction} satisfies~\eqref{sim2} with $\alpha = \ul\mu	\ul\alpha$.
		For any $s\in S$, $s_i\in S_i$, $x\in X$, $\hat x\in \hat X$, and taking $b=p$, it follows from~\eqref{sim_eq_viest} and Assumption~\ref{ass_external_gains} that
		\begin{align*}
		\sum_{i=1}^{\infty} \mu_i V_{i, s_i}(x_i, \hat x_i) &\geq \sum_{i=1}^{\infty} \mu_i \alpha_i|h_{i, s_{i}}(x_i) - \hat h_{i, s_{i}}(\hat x_i)|^p\\
		& \geq \ul\mu	\ul\alpha \sum_{i=1}^{\infty} |h_{i, s_{i}}(x_i) - \hat h_{i, s_{i}}(\hat x_i)|^p \\
		& \geq \ul\mu	\ul\alpha |h_s(x) - \hat h_s(\hat x)|_p^p  .
		\end{align*}
		Now we show the inequality \eqref{sim3} holds as well. Considering \eqref{sim_eq_nablaviest} and \eqref{eq:sim-function-construction}, we obtain the chain of inequalities in \eqref{eq_nablaviest} for all $s'_i, s_i\in S_i$, $s_j\in S_j$, $s', s\in S$, $i\in \N$.
		\begin{figure*}[ht]
			\rule{\textwidth}{0.4pt}
			\begin{align}\label{eq_nablaviest}
			\begin{array}{l}
			V_{s'}\left(f_{s}(x, u),\hat f_{s}(\hat x, \hat u)\right)- V_{s}(x,\hat x) = \sum\limits_{i = 1}^\infty  \mu_{i}\left[{V_{i, s'_i}}\left(f_{i, s_i}(x_i, w_i, u_i),\hat f_{i, s_i}(\hat x_i, \hat w_i, \hat u_i)\right)- {V_{i, s_i}}(x_i,\hat x_i) \right]
			\\\le \sum\limits_{i = 1}^\infty  {\mu_{i}}( - {\lambda_i}{V_{i, s_i}}(x_i, \hat x_i) + {\rho _{i,{\mathop{\rm int}} }}\left| {{w_i} - {\hat w_i}} \right|^p + {\rho _{i,\rm ext}}|\hat u_i|^{q})
			\\\le \sum\limits_{i = 1}^\infty  {\mu _i}( - {\lambda_i}{V_{i, s_i}}(x_i, \hat x_i) + \sum\limits_{j \in {I_i}} {{\rho _{i,{\mathop{\rm int}} }}{\bar N_i}\left| {{w_{ij}} - {\hat w_{ij}}} \right|^p}  + {\rho _{i,\rm ext}}|\hat u_i|^{q}) \\
			\le \sum\limits_{i = 1}^\infty  {\mu_{i}}( - {\lambda_i}{V_{i, s_i}}(x_i, \hat x_i) + \sum\limits_{j \in {I_i}} {{\rho _{i,{\mathop{\rm int}} }}{\bar N_i}\left| {{h_{j, s_j}}({x_j}) - {{\hat h}_{j, s_j}}({{\hat x}_j})} \right|^p} + {\rho _{i,\rm ext}}|\hat u_i|^{q}) \\
			\le \sum\limits_{i = 1}^\infty  {\mu_{i}}( - {\lambda_i}{V_{i, s_i}}({x_i},{{\hat x}_i}) + \sum\limits_{j \in {I_i}} {{\rho _{i,{\mathop{\rm int}} }}{\bar N_i}\frac{1}{\alpha _j}{V_{j, s_j}}({x_j},{{\hat x}_j})} + {\rho _{i,\rm ext}}|\hat u_i|^{q}) \\
			\mathop  \leq \limits^{(\ref{gamma})} \sum\limits_{i = 1}^\infty  {\mu_{i}}\left( - {\lambda_i}{V_{i, s_i}}({x_i},{{\hat x}_i}) + \sum\limits_{j \in {I_i}} {{\gamma_{ij}}{V_{j, s_j}}({x_j},{{\hat x}_j}})  + {\rho _{i,\rm ext}}|\hat u_i|^{q}\right).
			\end{array}
			\end{align}
			\rule{\textwidth}{0.4pt}
		\end{figure*}
		
		Letting $V_{s_{vec}}(x,\hat x) := \left(V_{i,s_i}(x_i, \hat x_i)\right)_{i\in\N}$ and using~\eqref{eq_nablaviest} and~\eqref{eq_muest2} , we have that	
		\begin{align*}
		&V_{s'}(f_{s}(x, u),\hat f_{s}(\hat x, \hat u))- V_{s}(x,\hat x) \nonumber\\&\leq  \Bigl[ \mu\trn(-\Lambda + \Gamma)V_{s_{vec}}(x,\hat x) + \sum_{i=1}^{\infty} {\mu_{i}}{\rho _{i,\rm ext}}|\hat u_i|^{q}\Bigr]\nonumber\\&
		\leq -\lambda_{\infty} V_{s}(x,\hat x) + {\rho _{\rm ext}}(|\hat u|_q), 
		\end{align*}
		where ${\rho _{\rm ext}}(t)= \overline{\mu}\;{{\overline\rho _{\rm{ext}}}}t^q$ for all $t\geq 0$.
	\end{proof}
	
	\subsection{Abstractions for linear systems}
	
	In this section, we use the previous results to compute the compositional abstractions for a network of linear switched subsystems.  
	We aim to construct abstractions with output trajectories close enough to those of the concrete system.
	
	Consider the following network on interconnected linear switched systems:
\begin{equation}\label{eq_ith_subsystem_lin}
		{\Sigma _i}: \left\{ \begin{array}{l}
		\mathbf x_i(k+1)  = A_{i, \sigma_i(k)}{\mathbf x}_i(k)+D_{i, \sigma_i(k)}{\mathbf w_i(k)}\\\;\;\;\;\;\;\;\;\;\;\;\;\;\;\;\;\;\;\;+B_{i, \sigma_i(k)}{\mathbf u}_i(k) ,\\
		{\mathbf y_i(k)} = C_{i, \sigma_i(k)}{\mathbf x_i(k)},
		\end{array} \right.%
		\end{equation}
	where $\sigma_i\in{\cal{S}}_i$, $A_{i, \sigma_i(k)}\in \R^{n_i\times n_i}$, $B_{i, \sigma_i(k)}\in \R^{n_i\times m_i}$, $C_{i, \sigma_i(k)}\in \R^{q_i\times n_i}$ and $D_{i, \sigma_i(k)}\in \R^{n_i\times p_i}$ for $i\in\N$.
	
	Choose $ X = \ell^2(\N,(n_i))$ and $U = \ell^2(\N,(m_i))$. We assume that we are given abstractions as $\hat\Sigma_i$, and then provide conditions under which $V_{i, s_i}$ for $s_i\in S_i$ are candidate simulation functions from $\hat\Sigma_i$ to $\Sigma_i$. Assume that there exist a family of matrices $K_{i, s_i}$, positive definite matrices $M_{i, s_i}$ and given $0<\kappa_i<1$ for $i\in \N$, such that the following matrix inequalities hold for all consecutive $s'_i,s_i\in S_i$ (i.e., $s'_i=\sigma_i(k+1), s_i=\sigma_i(k)$ for $k\in \N_0$).
	\begin{subequations}\label{cond22}
		\begin{align}
		&C_{i, s_i}^\top{C_{i, s_i}} \preceq {M_{i, s_i}},\label{cond222}\\&
		3{({A_{i, s_i}} + {B_{i, s_i}}{K_{i, s_i}})^\top}{M_{i, s'_i}}({A_{i, s_i}} + {B_{i, s_i}}{K_{i, s_i}}) - {M_{i, s_i}} \nonumber\\&\preceq  - \kappa_i M_{i, s_i}.\label{cond2222}
		\end{align}
	\end{subequations}
	Note that \eqref{cond2222} could be transformed to a LMI using the Schur complement lemma (see~\cite[Remark 4.7]{swikir2019compositional}).
	
	Take the following simulation function candidates for the mentioned chosen state space:
	\begin{align}\label{sim_func}
	{V_{i, s_i}}({x_i},{{\hat x}_i}) = {({x_i} - {P_{i}}{{\hat x}_i})^\top}{M_{i, s_i}}({x_i} - {P_{i}}{{\hat x}_i}).
	\end{align}
The input is given by the interface function $\nu_i$ as follows.
		\begin{align}\label{interface}
		u_i&=\nu_i (x_i,{\hat x}_i,{\hat u}_i,{\hat w}_i,s_i) \\\notag
		= & K_{i, s_i}({x_i} - P_i {\hat x}_i) + Q_{i, s_i} {\hat x}_i+ R_{i, s_i}{\hat u}_i + T_{i, s_i} {\hat w}_i ,
		\end{align}
		where $P_{i}$, $i\in\N$, are appropriate dimension matrices.
	Assume that the following inequalities hold for some appropriate dimension matrices $Q_{i, s_i}$, $T_{i, s_i}$.
	\begin{subequations}\label{cond2}
		\begin{align}
		&A_{i, s_i} P_{i} =  P_{i} {\hat A}_{i, s_i} - B_{i, s_i} Q_{i, \sigma_i} , \label{cond2_1}\\
		&{D_{i, s_i}} = {P_{i}}{{\hat D}_{i, s_i}} - {B_{i, s_i}}{T_{i, \sigma_i}},\label{cond2_2}\\
		&{C_{i, s_i}}{P_{i}} = {{\hat C}_{i, s_i}}.\label{cond2_3}
		\end{align}
	\end{subequations}
	\begin{theorem}
		Consider two systems $\Sigma_i=({A_{i, s_i}}, {B_{i, s_i}},$ ${C_{i, s_i}}, {D_{i, s_i}})$ and $\hat \Sigma_i=({\hat A_{i, s_i}}, {\hat B_{i, s_i}}, {\hat C_{i, s_i}}, {\hat D_{i, s_i}})$ for $i\in\N$. Suppose that for all $s_i\in S_i$ there exist appropriate matrices $M_{i, s_i}$, $P_{i}$, $K_{i, s_i}$, $Q_{i, s_i}$ and $T_{i, s_i}$ which satisfy \eqref{cond22} and \eqref{cond2}. Then, the functions defined in \eqref{sim_func}  are simulation functions from $\hat \Sigma_i$ to $\Sigma_i$ with inputs given by \eqref{interface}.
	\end{theorem}
	\begin{proof}
		According to \eqref{cond2_3}, we have 
		\begin{align*}
		| C_{i, s_i}x_i \!- \!\hat{C}_{i, s_i}&\hat{x}_i| =\\
		&\big((x_i\!-\!P_{i}\hat C_{i, s_i})^\top \!C_{i, s_i}^\top \!C_{i, s_i}(x_i\!-\!P_{i}\hat C_{i, s_i})\big)^\frac{1}{2}.
		\end{align*}
			Using \eqref{cond222}, it is clear that $| {C_{i, s_i}x_i - \hat C_{i, s_i}\hat x_i}|^2 \le V_{i, s_i}({x_i},{{\hat x}_i})$ holds for all $x_i \in \R^{n_i}$, $\hat x_i \in \R^{\hat n_i}$. Then, \eqref{sim_eq_viest} is satisfied with $\alpha_i=1$, $i\in\N$, $p=2$.
			
			Now, we proceed to show that \eqref{sim_eq_nablaviest} is satisfied, too.
			
			By using \eqref{cond2_1}, \eqref{cond2_2} and considering the $u_i$ given by \eqref{interface}, $A_{i, s_i}x_i + B_{i, s_i}u_{i} + D_{i, s_i}w_i - P_{i}(\hat A_{i, s_i}\hat x_i + \hat B_{i, s_i}\hat u_i + \hat D_{i, s_i}\hat w_i)$ is simplified to $(A_{i, s_i} + B_{i, s_i}K_{i, s_i})(x_i - P_{i}\hat x_i) + D_{i, s_i}(w_i - \hat w_i) + (B_{i, s_i}R_{i, s_i} - P_{i}\hat B_{i, s_i})\hat u_i$.
		Therefore, we obtain 
		\begin{align}\label{3rd cond}
		V_{i,s_i'}&\left(f_{i,s_i}(x_i, w_i, u_{i}),\hat f_{i,s_i}(\hat x_i, \hat w_i, \hat u_i)\right)- V_{i,s_i}(x_i,\hat x_i)\nonumber\\
		=&   {({x_i} - {P_{i}}{{\hat x}_i})^\top}[{({A_{i, s_i}} + {B_{i, s_i}}{K_{i, s_i}})^\top}{M_{i, s'_i}} \nonumber\\
		&\times({A_{i, s_i}} + {B_{i, s_i}}{K_{i, s_i}}) - {M_{i, s_i}}]{({x_i} - {P_{i}}{{\hat x}_i})}\nonumber\\&
		+ [2{({x_i} - {P_{i}}{{\hat x}_i})^\top}{({A_{i, s_i}} + {B_{i, s_i}}{K_{i, s_i}})^\top}]{M_{i, s'_i}}\nonumber\\
		&\times[{D_{i, s_i}}(w_i-\hat w_i)]+[2{({x_i} - {P_{i}}{{\hat x}_i})^\top}\nonumber\\
		&\times{({A_{i, s_i}} + {B_{i, s_i}}{K_{i, s_i}})^\top}]{M_{i, s'_i}}[(B_{i, s_i}R_{i, s_i}-P_{i}\hat B_{i, s_i})\hat u_i]\nonumber\\&+[2(w_i-\hat w_i)^\top {D_{i, s_i}}^\top]{M_{i, s'_i}}[(B_{i, s_i}R_{i, s_i}-P_{i}\hat B_{i, s_i})\hat u_i]\nonumber\\&+|\sqrt{M_{i, s'_i}}D_{i, s_i}(w_i-\hat w_i)|^2\nonumber\\&+|\sqrt{M_{i, s'_i}}(B_{i, s_i}R_{i, s_i}-P_{i}\hat B_{i, s_i})\hat u_i|^2. 
		\end{align}
		Using Young's inequality and \eqref{cond2222}, we can obtain the following:
		\begin{align*}\label{4th cond}
		&V_{i,s_i'}\left(f_{i,s_i}(x_i, w_i, u_i),\hat f_{i,s_i}(\hat x_i, \hat w_i, \hat u_i)\right)- V_{i,s_i}(x_i,\hat x_i)\leq\nonumber\\&   - \kappa_i V_{i,s_i}(x_i,\hat x_i) +3|\sqrt{M_{i, s'_i}}D_{i, s_i}|^2|w_i-\hat w_i|^2\nonumber\\
		&+3\Big|\sqrt{M_{i, s'_i}}(B_{i, s_i}R_{i, s_i}-P_{i}\hat B_{i, s_i})\Big|^2|\hat u_i|^2. 
		\end{align*}
			Thus, \eqref{sim_eq_nablaviest} holds with $p=q=2$, $\lambda_i=\kappa_i$, ${\rho _{\rm{i,ext}}}=3\mathop {\max }\limits_{{s_i}} \{|\sqrt{M_{i, s'_i}}(B_{i, s_i}R_{i, s_i}-P_{i}\hat B_{i, s_i})|^2 \}$ and ${\rho _{{\mathop{i,\rm int}} }}=3\mathop{\max }\limits_{{s_i}} \{|\sqrt{M_{i, s'_i}}D_{i, s_i}|^2\}$.
			
			Therefore, the functions of \eqref{sim_func} are simulation functions from $\hat\Sigma_i$ to $\Sigma_i$.
	\end{proof}
	\section{Example}\label{sec:Example}
	We verify the effectiveness of our theoretical results by regulating the frequency deviations in a power network.
	
	We consider a power network modeled by an interconnection of second-order systems, known as swing equation~\cite{Kundur.1994}.
	In particular, we consider two \emph{circular} topologies: in the first one shown in Figure~\ref{block1}, subsystem $i$ is fed by subsystem $i-1$; in the other configuration, subsystem $i$ is fed by subsystem $i+1$, see Figure~\ref{block2}.
	We assume that the network topology switches between these two configurations at certain times.
	Let ${\sigma_i(k)}$ be the switching signal which takes values in the set $\{1,2\}$, where $\sigma_i(k) = 1$ corresponds to the topology shown in Figure~\ref{block1} and $\sigma_i(k) = 2$ corresponds to that illustrated in Figure~\ref{block2}.
	To mathematically describe such a relation between the network topology and the switching signal, we define the function $g_i$ by
	\begin{align*}
	g_i (s) = \left\{\begin{array}{lcl}
	i - 1 & \textrm{if } & s = 1 , \\
	i + 1 & \textrm{if } & s = 2 .
	\end{array}\right.
	\end{align*}
	In that way, each subsystem of the network is described by
	\begin{align*}
	\Sigma _i:\!\left\{ \begin{array}{l}
	\left[ \begin{array}{*{20}{c}}
	\delta_i (k + 1) \\
	\omega_i (k + 1)
	\end{array} \right] \!\!=\!\! \underbrace{\left[ \begin{array}{*{20}{c}}
		1&1\\
		\frac{{ - {l_{i g_{i}(\sigma_i(k))}} }}{{{m_i}}}&{1 - \frac{{{d_i}}}{{{m_i}}}}
		\end{array} \right]}_{=:A_{i,{\sigma_i (k)}}} \underbrace{\left[ \begin{array}{*{20}{c}}
		\delta _i(k)\\
		\omega _i(k)
		\end{array} \right]}_{=:\mathbf x_i (k)} \\
	\;\;\;\;\;\;\;\;\;\;\;\;\;\;\;\;\;\;\;\;\;\;\;
	+  \underbrace{\left[ {\begin{array}{*{20}{c}}
			0 \\
			\frac{l_{i g_{i}(\sigma_i(k))}}{m_i}
			\end{array}} \right]}_{=:D_{i,{\sigma_i(k)}}} \,\,\underbrace{\delta_{g_{i}(\sigma_i(k))} (k)}_{=: \mathbf w_i (k)} \\
	\;\;\;\;\;\;\;\;\;\;\;\;\;\;\;\;\;\;\;\;\;\;\;
	+ \underbrace{\left[ {\begin{array}{*{20}{c}}
			0\\
			\frac{1}{m_i}
			\end{array}} \right]}_{=:B_{i}} {\mathbf u}_{i} (k),\\
	{\mathbf y_i}(k) = \underbrace{\left[ \begin{array}{*{20}{c}}
		C_{ii}\\
		C_{ig_{i}(\sigma_i(k))}
		\end{array} \right]}_{=:C_{i,{\sigma_i (k)}}} \left[ \begin{array}{*{20}{c}}
	\delta_i(k)\\
	\omega_i(k)
	\end{array} \right],
	\end{array} \right.
	\end{align*}
where $C_{ii}\!=\!\begin{bmatrix}
	0& 1
	\end{bmatrix}\!,
	C_{ig_{i}(\sigma_i(k))}\!=\!\begin{bmatrix}
	1& 0
	\end{bmatrix}$, and $\delta_i$, $\omega_i$, $m_i$, $d_i$, and $\mathbf u_{i}$ are the phase angle, frequency, inertia, damping coefficient, and the mechanical input power of bus $i$, respectively.
	The coefficient $l_{i g_i (s_i)}=|v_i||v_{g_i (s_i)}|b_{i g_i (s_i)}$, where $|v_i|$ is the absolute value of the voltage of bus $i$, and $b_{i s_i}$ is the susceptance of the line $(i,g_i (s_i))$ for $s_i \in \{1,2\}$.
	
	The interconnection structure switches between two \emph{circular} topologies shown in Figures~\ref{block1}-\ref{block2}.
	\begin{figure}
		\centering  
		\includegraphics[width=6.5cm]{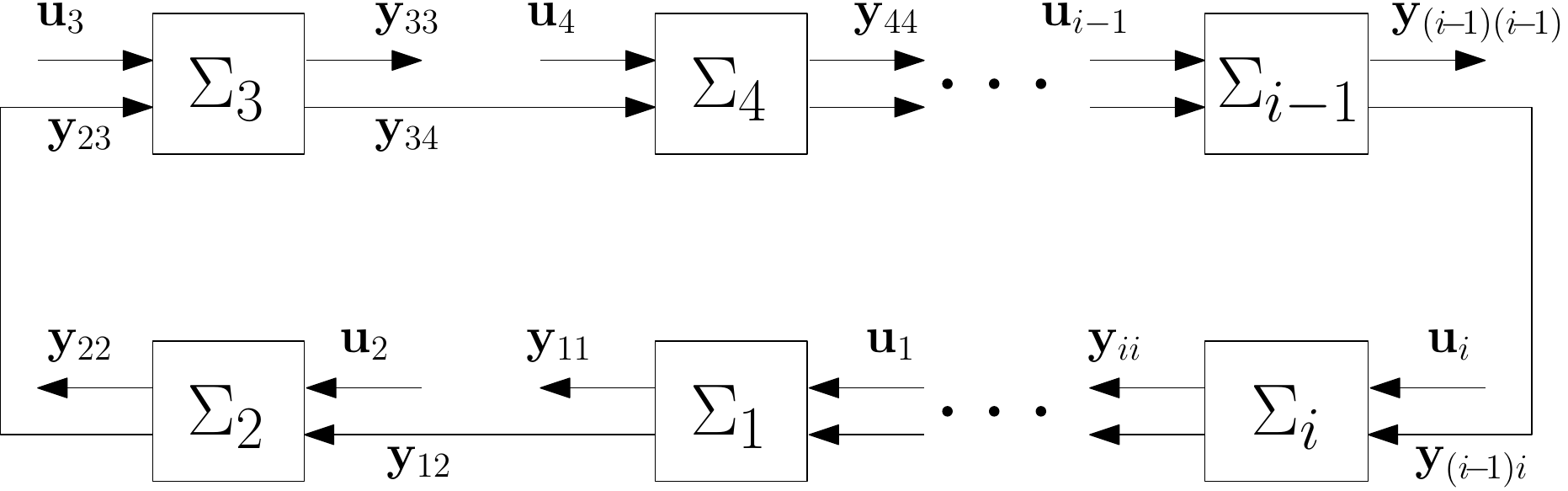}
		\caption{The interconnected system $\Sigma$ for $s_i=1$.}\label{block1}
	\end{figure}
		
	\begin{figure}
		\centering  
		\includegraphics[width=6.5cm]{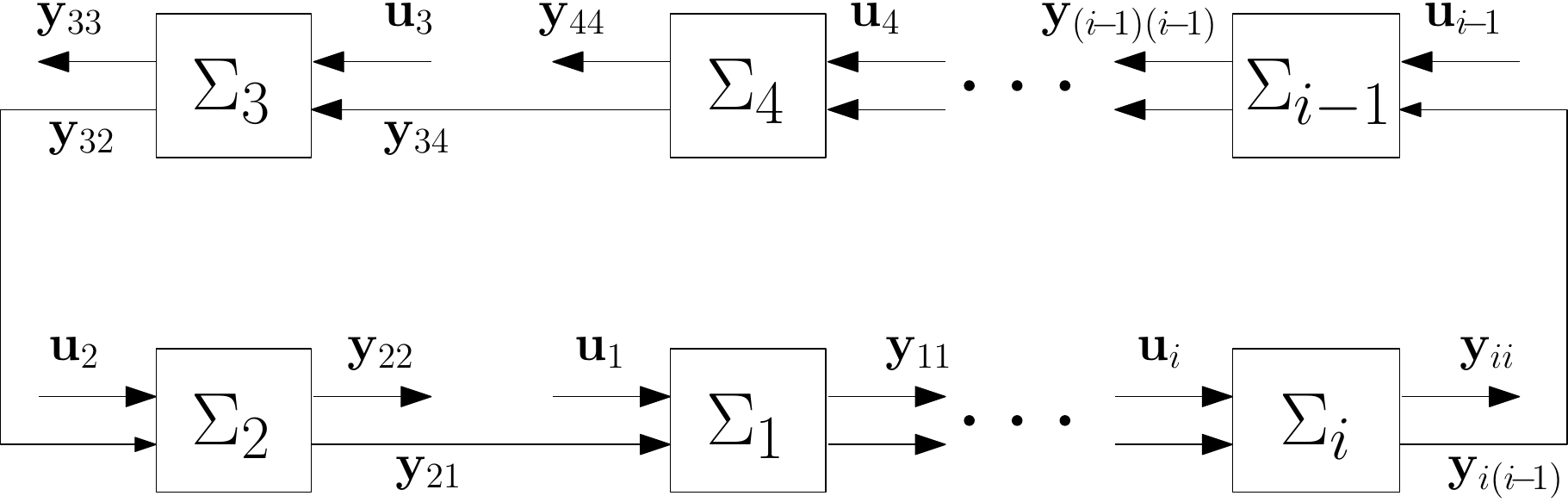}
		\caption{The interconnected system $\Sigma$ for $s_i=2$.}\label{block2}
	\end{figure}
To construct abstractions for $\Sigma$, we construct an abstraction for $\Sigma_i$, $i\in \N$, for both communication topologies, i.e. for both $s_i =1,2$.
Taking $K_{i,s_i}=\left[ \begin{array}{*{20}{c}}
l_{i g_i(s_i)}-\frac{9}{16} m_i &{d_i} - 1.5 m_i
\end{array} \right]$, for both $s_i =1,2$, and $\kappa_i=0.2$, we compute $M_{i,{s'_i}}=M_{i,{s_i}}=\left[ {\begin{array}{*{20}{c}}
			11.20 &12.50\\
			12.50 &17.83
\end{array}} \right]$, for $s'_i,s_i\in \{1,2\}$, which satisfies~\eqref{cond22}.
Now we proceed to compute other matrices so that~\eqref{cond2} holds.
Using~\eqref{cond2_1}, we take $Q_{i,s_i}= l_{ig_{i}(s_i)}$ for $s_i=1,2$.
We obtain $\hat A_{i,{s_i}}=c_{i}$, $s_i=1,2$, and $P_{i} = \left[ {1;c_{i} - 1} \right]$ for constant $c_{i}$ which is determined by solving the equation $c_{i}^2+c_{i}(\frac{d_i}{m_i}-2)+1=0$. Therefore, $c_i=1-\frac{d_i}{2m_i}$.
Moreover, using~\eqref{cond2_2}, we get $\hat D_{i,{s_i}}=0$ and ${T_{i,{s_i}}} = -{l_{i g_i(s)}}$.
Accordingly, $\hat C_{i, s_i}=C_{i, s_i}\left[ {\begin{array}{*{20}{c}}
1\\
c_i-1
\end{array}} \right]$. 

We also choose $\hat B_{i, s_i}=\frac{d_i}{2m_i}-0.6$ and $R_{i, s_i}=(B_{i, s_i}^\top M_{i, s_i}B_{i, s_i})^{-1}B_{i, s_i}^\top M_{i, s_i}P_{i}\hat B_{i, s_i}$ to minimize $\rho_{\rm{i,ext}}$ as suggested in~\cite{girard2009hierarchical}.
	
	With the choice of $V_i$, Assumptions \ref{SIM_vi_existence} and~\ref{ass_external_gains} hold with $\alpha_i=1$, $\lambda_i=0.2$, $\rho_{i,\rm int}=0.1455$, $\rho_{\rm{i,ext}}=8.1487\times 10^{-11}$.
	Recalling the circular interconnection topologies, each subsystem is directly fed by either subsystem $i-1$ or $i+1$ at each time instant.
	So \eqref{gamma} gives $\gamma_{ij}=3\mathop{\max }\limits_{{s_i}} \{|\sqrt{M_{i, s'_i}}D_{i, s_i}|^2\}{\bar N_i}\frac{1}{\alpha_j}=0.1455$ for $j\in I_i$ and $\gamma_{ij}=0$ for $j \notin {I_i}$.
	Then we get
	\begin{align*}
	r(\Psi) <\sup_{j \in \N} \sum_{i=1}^{\infty} \psi_{ij}<\frac{0.1455}{0.2}<1, 
	\end{align*}
	which implies that the spectral radius condition~\ref{ass:spectral-radius} is fulfilled.
	Therefore all the hypotheses of Theorem~\ref{MTC} are satisfied.
	
	For the sake of simulations, we consider a network of $1000$ subsystems.
	The parameter values are set as $m_i=10^5\rm{kgm^2}$, $d_i=1 s^{-1}$, $l_{ii}=4\times10^3$ for all $i\in\left\{1,\cdots,1000\right\}$.
	Additionally for $s_i =1$ and $s_i =2$, we have $l_{i (i-1)}=4\times10^3$ and $l_{i (i+1)}=4\times10^3$, respectively.
	Recalling the computed matrices $\hat A_{i, s_i}$ and $\hat B_{i, s_i}$, by taking each local controller of the abstract subsystem as $\hat u_i=\hat x_i$ the network $\hat\Sigma$ gets stabilized at the origin.
	The switching between $s_i=1$ and $s_i=2$ occurs at $k=5n,$ $n\in \N$ time instants.
	The norm of the overall error between the output trajectories of the \emph{abstract} and \emph{concrete} systems and the closed-loop output trajectories of the \emph{concrete} subsystems are, respectively, depicted by Figures~\ref{xx} and~\ref{xx2}.
	From the choice of $\hat u$ and stabilizability of $\hat\Sigma$ at the origin, $\lim_{k\to\infty} |\hat{\mathbf u}(k)|_2 \to 0$.
	This together with~\eqref{sim3} implies that the mismatch between output trajectories converges to zero, which is illustrated by Fig.~\ref{xx}.
	\begin{figure}
		\centering
		\hspace*{-0.5cm}
		\includegraphics[width=9cm]{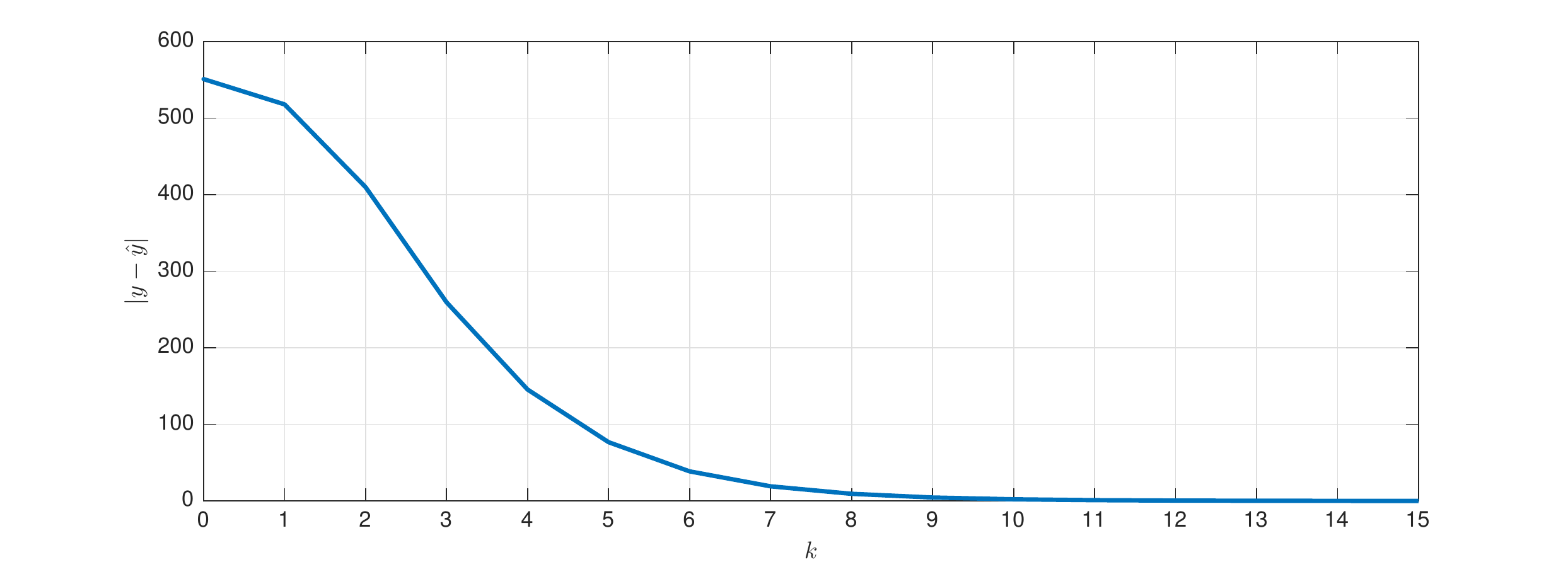}
		\caption{The error norm between the output trajectories of $\Sigma$ and $\hat \Sigma$, consisting of $1000$ subsystems.}
		\label{xx}
	\end{figure}
	\begin{figure}
		\centering
		\hspace*{-0.5cm}
		\includegraphics[width=9cm]{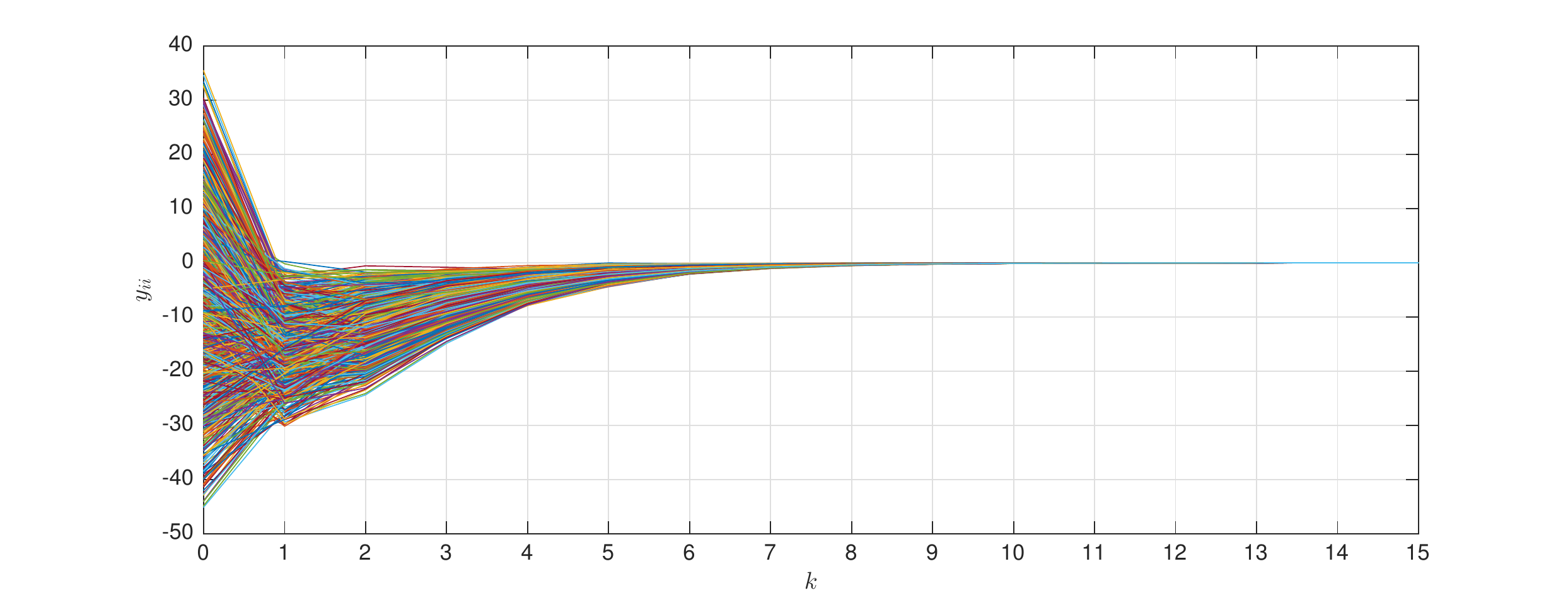}
		\caption{The external outputs $\mathbf y_{ii}$ (i.e. frequency deviations) for $i={1,\dots,1000}$.}
		\label{xx2}
	\end{figure}

	\section{Conclusions}\label{sec:Conclusions}
	
	We constructed continuous abstractions compositionally for an infinite network of switched discrete-time systems with arbitrary switching signals.
	To do this,	we extended the notion of simulation functions to infinite-dimensional systems (networks of infinitely many finite-dimensional systems).
	Following the compositionality approach, we assigned to each subsystem an individual simulation function and constructed each local abstraction accordingly.
	Finally we composed the local abstractions to provide an abstraction of the overall network.
	We showed that the aggregation yields a continuous abstraction of the overall concrete network if the small-gain condition, expressed in terms of a spectral radius criterion, is satisfied.
	For linear systems, our approach boils down to linear matrix inequality conditions which can be computed efficiently.
	We applied our result to a power network with a switched topology.
	
	\bibliographystyle{IEEEtran}
	\bibliography{references}

\begin{thebibliography}{10}
\providecommand{\url}[1]{#1}
\csname url@samestyle\endcsname
\providecommand{\newblock}{\relax}
\providecommand{\bibinfo}[2]{#2}
\providecommand{\BIBentrySTDinterwordspacing}{\spaceskip=0pt\relax}
\providecommand{\BIBentryALTinterwordstretchfactor}{4}
\providecommand{\BIBentryALTinterwordspacing}{\spaceskip=\fontdimen2\font plus
\BIBentryALTinterwordstretchfactor\fontdimen3\font minus
  \fontdimen4\font\relax}
\providecommand{\BIBforeignlanguage}[2]{{%
\expandafter\ifx\csname l@#1\endcsname\relax
\typeout{** WARNING: IEEEtran.bst: No hyphenation pattern has been}%
\typeout{** loaded for the language `#1'. Using the pattern for}%
\typeout{** the default language instead.}%
\else
\language=\csname l@#1\endcsname
\fi
#2}}
\providecommand{\BIBdecl}{\relax}
\BIBdecl

\bibitem{pt09}
G.~Pola and P.~Tabuada, ``Symbolic models for nonlinear control systems:
  Alternating approximate bisimulations,'' \emph{SIAM Journal on Control and
  Optimization}, vol.~48, no.~2, pp. 719--733, 2009.

\bibitem{girard2009hierarchical}
A.~Girard and G.~J. Pappas, ``Hierarchical control system design using
  approximate simulation,'' \emph{Automatica}, vol.~45, no.~2, pp. 566--571,
  2009.

\bibitem{Smith.2018}
S.~W. {Smith}, M.~{Arcak}, and M.~{Zamani}, ``Hierarchical control via an
  approximate aggregate manifold,'' in \emph{Amer. Control Conf.}, 2018, pp.
  2378--2383.

\bibitem{Smith.2019}
S.~W. {Smith}, H.~{Yin}, and M.~{Arcak}, ``{Continuous Abstraction of Nonlinear
  Systems using Sum-of-Squares Programming},'' \emph{arXiv e-prints}, 2019.

\bibitem{rungger2016compositional}
M.~Rungger and M.~Zamani, ``Compositional construction of approximate
  abstractions of interconnected control systems,'' \emph{IEEE Trans. Control
  Netw. Syst.}, vol.~5, no.~1, pp. 116--127, 2016.

\bibitem{ZamaniArcak2017}
M.~Zamani and M.~Arcak, ``Compositional abstraction for networks of control
  systems: A dissipativity approach,'' \emph{IEEE Trans. Control Network
  Syst.}, vol.~5, no.~3, pp. 1003--1015, 2017.

\bibitem{Noroozi.2018b}
N.~Noroozi, F.~R. Wirth, and M.~Zamani, ``Compositional construction of
  abstractions via relaxed small-gain conditions {Part I}: continuous case,''
  in \emph{Euro. Control Conf.}, Limassol, June 2018, pp. 76--81.

\bibitem{NMK20b}
N.~Noroozi, A.~Mironchenko, C.~Kawan, and M.~Zamani, ``Small-gain theorem for
  stability, cooperative control, and distributed observation of infinite
  networks,'' \emph{Submitted, see also http://arxiv.org/abs/2002.07085}, 2020.

\bibitem{DaP20}
S.~Dashkovskiy and S.~Pavlichkov, ``Stability conditions for infinite networks
  of nonlinear systems and their application for stabilization,''
  \emph{Automatica}, vol. 112, p. 108643, 2020.

\bibitem{Bamieh.2012}
B.~{Bamieh}, M.~R. {Jovanovic}, P.~{Mitra}, and S.~{Patterson}, ``Coherence in
  large-scale networks: Dimension-dependent limitations of local feedback,''
  \emph{IEEE Trans. Autom. Control}, vol.~57, no.~9, pp. 2235--2249, 2012.

\bibitem{Barooah.2009}
P.~{Barooah}, P.~G. {Mehta}, and J.~P. {Hespanha}, ``Mistuning-based control
  design to improve closed-loop stability margin of vehicular platoons,''
  \emph{IEEE Trans. Autom. Control}, vol.~54, no.~9, pp. 2100--2113, 2009.

\bibitem{BPD02}
B.~Bamieh, F.~Paganini, and M.~A. Dahleh, ``Distributed control of spatially
  invariant systems,'' \emph{IEEE Trans. Autom. Control}, vol.~47, no.~7, pp.
  1091--1107, 2002.

\bibitem{Jovanovic.2005b}
M.~R. Jovanovi\'{c} and B.~Bamieh, ``On the ill-posedness of certain vehicular
  platoon control problems,'' \emph{IEEE Trans. Autom. Control}, vol.~50,
  no.~9, pp. 1307--1321, 2005.

\bibitem{kawan2019lyapunov}
C.~Kawan, A.~Mironchenko, A.~Swikir, N.~Noroozi, and M.~Zamani, ``A
  {L}yapunov-based {ISS} small-gain theorem for infinite networks,''
  \emph{Submitted, see also: http://arxiv.org/abs/1910.12746}, 2019.

\bibitem{Khalil.2002}
H.~K. Khalil, \emph{Nonlinear systems}, 3rd~ed.\hskip 1em plus 0.5em minus
  0.4em\relax Englewood Cliffs, NJ: Prentice-Hall, 2002.

\bibitem{Jiang.2001}
Z.-P. Jiang and Y.~Wang, ``Input-to-state stability for discrete-time nonlinear
  systems,'' \emph{Automatica}, vol.~37, no.~6, pp. 857--869, 2001.

\bibitem{Runggerhscc}
M.~Rungger and M.~Zamani, ``Compositional construction of approximate
  abstractions,'' in \emph{18th Int. Conf. Hybrid Syst. Computation Control},
  New York, 2015, pp. 68--77.

\bibitem{swikir2019compositional}
A.~Swikir and M.~Zamani, ``Compositional synthesis of finite abstractions for
  networks of systems: A small-gain approach,'' \emph{Automatica}, vol. 107,
  pp. 551--561, 2019.

\bibitem{Kundur.1994}
P.~Kundur, N.~J. Balu, and M.~G. Lauby, \emph{Power system stability and
  control}.\hskip 1em plus 0.5em minus 0.4em\relax McGraw-hill New York, 1994,
  vol.~7.

\end{thebibliography}

\end{document}